\documentclass{article}
\usepackage{amssymb}
\usepackage{amsmath,amsfonts,amsthm,amscd,upref,amstext,mathtools,comment}
\usepackage{tikz}
\usetikzlibrary{trees}
\usepackage{graphicx}
\usepackage{caption}
\usepackage{subcaption}
\usepackage{cancel}
\usepackage[colorlinks=true,urlcolor=blue, citecolor=red,linkcolor=blue,linktocpage,pdfpagelabels, bookmarksnumbered,bookmarksopen]{hyperref}
\usepackage[hyperpageref]{backref}
\usepackage[noabbrev]{cleveref}
\allowdisplaybreaks

\newcommand{\Greedy}{\mathrm{GREEDY}}
\newcommand{\OPT}{\mathrm{OPT}}
\newcommand{\ON}{\mathrm{ON}}

\newtheorem{theorem}{Theorem}
\newtheorem{lemma}[theorem]{Lemma}

\theoremstyle{definition}

\newtheorem{definition}[theorem]{Definition}

\numberwithin{theorem}{section} \numberwithin{equation}{section}

\usepackage{xcolor}  

\usepackage[colorinlistoftodos,prependcaption,textsize=tiny]{todonotes}

\usepackage[colorlinks=true,urlcolor=blue, citecolor=red,linkcolor=blue,linktocpage,pdfpagelabels, bookmarksnumbered,bookmarksopen]{hyperref}
\usepackage[hyperpageref]{backref}
\usepackage{cleveref}

\newenvironment{keyword}{\begin{flushleft}\textbf{KEYWORDS}\\}{\end{flushleft}}

\begin{document}

\title{Resource Augmentation Analysis of the Greedy Algorithm for the Online Transportation Problem}

\author{Stephen Arndt \thanks{Computer Science Department, University of Pittsburgh, Pittsburgh PA, 15260. sda19@pitt.edu} \and Josh Ascher\thanks{Computer Science Department, University of Pittsburgh, Pittsburgh PA, 15260. joa71@pitt.edu} \and Kirk Pruhs\thanks{Computer Science Department, University of Pittsburgh, Pittsburgh PA, 15260. kirk@cs.pitt.edu\\Supported in part by NSF grants  CCF-1907673,  CCF-2036077, CCF-2209654 and an IBM Faculty Award.}}

\maketitle

\begin{abstract}
We consider the online transportation problem set in a metric space containing parking garages of various capacities.
Cars arrive over time, and must be assigned to
an unfull parking garage upon their arrival. The objective is to minimize the aggregate distance
that cars have to travel to their assigned parking garage. We show that the natural greedy
algorithm, augmented with garages of $k\ge3$ times the capacity, is $\left(1 + \frac{2}{k-2}\right)$-competitive.
\end{abstract}


\section{Introduction}

  We consider the natural online version of the classical transportation problem \cite{kenningtonalgsfornetworks, lawler1976combinatorial}. The 
setting is a metric space $\mathcal{M}$ that contains a collection $S = \{ s_1, s_2, \dots, s_m \}$ of server sites at various locations in $\mathcal{M}$. Each server site $s_j$ has a positive integer capacity $a_j$. 
Conceptually think of each server site $s_j$ as a parking garage with $a_j$ parking spaces. 
Over time, a sequence of requests $R = \{ r_1, r_2, \dots, r_n \}$  arrive at various locations
in the metric space. Think of the requests as cars that are looking for a space to park.
Upon the arrival of each request $r_i$ the online algorithm $\mathcal{A}$ must assign $r_i$ to an unfull server site $s_{\sigma(i)}$, that is one where the
number of previous requests assigned to $s_{\sigma(i)}$ is less than $a_{\sigma(i)}$.
The  cost incurred by such an assignment is the distance $d(s_{\sigma(i)}, r_i)$ between the location of $s_{\sigma(i)}$ and 
the location where $r_i$ arrived in $\mathcal{M}$. The objective  is to minimize $\sum_{i=1}^n d(s_{\sigma(i)}, r_i)$, the total cost to service the requests. 
So in our parking application, the objective would be to minimize the aggregate
distance that the cars have to travel to reach their assigned parking space.
In this setting, one standard performance metric of an online algorithm  is the competitive ratio. 
An online algorithm $\mathcal{A}$ is $c$-competitive 
 if for all instances $I$ it is the case that $\mathcal{A}(I) \le c \cdot \OPT(I)$, 
 where $\mathcal{A}(I)$ is the objective value attained by the online algorithm $\mathcal{A}$ on
 instance $I$ and $\OPT(I)$ is the optimal objective value for instance $I$.

\subsection{The Essential Story So Far}

An important special case of the online transportation problem is the online metrical matching
problem, which is when each $a_i=1$. 
In \cite{onlineweightedmatching,khullermitchell} it was shown that the
optimal competitive ratio for online metrical matching 
is  $(2n - 1)$-competitive.
For online metrical matching the best known competitive ratio for a randomized algorithm against an oblivious adversary is $O(\log^2 n)$~\cite{meyerson,random-O(log2k)}, which is obtained
by an algorithm that uses a greedy algorithm on the embedding of the metric
space into a hierarchically separated tree (HST), and the best known lower bound
for the competitiveness of a randomized algorithm against an oblivious adversary is $\Omega(\log n)$.
Thus for online transportation no deterministic algorithm can be better than $2n-1$ competitive,
and no randomized algorithm can be $o(\log n)$-competitive.

The most natural algorithm for the online transportation problem is the greedy
algorithm $\Greedy$ that assigns each request to the nearest unfull server site. 
So understanding the performance of $\Greedy$, when 
it performs well and when it performs poorly, is of some interest.
In \cite{onlineweightedmatching} it was shown that  the competitive ratio of $\Greedy$ is $2^n - 1$, even for online metrical matching  in a line metric. 

One way to get around this strong worst-case lower bound for $\Greedy$ is to use resource
augmentation analysis. In this setting, this means assuming that for the online algorithm
the capacity $c_j$ of each server site $s_j$ is $c_j=k \cdot a_j$, where
$k$ is an integer strictly greater than one, while still assuming that in the benchmark optimal matching
the capacity of the garage is $a_j$. 
\cite{onlinetransportwa} showed that for all instances $I$,

$$\Greedy_2(I) \le O\left(\min(n, \log C) \cdot \OPT(I)\right)$$
where $\Greedy_2(I)$ is the objective value for $\Greedy$ assuming that
each server site $s_j$ has capacity $c_j = 2 a_j$, and $C=\sum_{j=1}^n c_j$
is the aggregate server capacity.
Further \cite{onlinetransportwa} showed how to modify the greedy algorithm,
by artificially increasing the distances to garages that are more than half full
by a constant multiplicative factor, to obtain an algorithm M$\Greedy$, and showed that

$$\mathrm{M}\Greedy_2(I) \le O\left(  \OPT(I)\right)$$
That is, this modified greedy has a constant competitive ratio if the capacity of its
server sites is doubled. 
\cite{extra-server} shows how to obtain an $O(\log^3 n)$-competitive randomized algorithm
 using HST's and resource augmentation of an additional
one server per site. 

Another way to get around the strong worst-case lower bound for $\Greedy$ is to use average-case analysis. 
\cite{tsaiagreedy} analyzes the average-case performance of $\Greedy$ for online
metrical matching in several natural
metric spaces. For example \cite{tsaiagreedy} shows that if the locations of the
requests and servers are uniformly and independently drawn from a Euclidean circle then in the limit as $n$ grows,
$$E[\Greedy(I)] \le 2.3 \sqrt{n} \cdot E[\OPT(I)]$$
As best as we can tell there are not results in the literature
on average-case analysis of $\Greedy$ for online metric matching or transportation in a general metric.

There are a significant number of papers that contain  (both average-case and worst-case) results for online metrical matching
and online transportation in metrics of special interest, most notably a line metric.
As our interests lie with general metric spaces, we will not survey these results here.

\subsection{Our Results}

Our main contribution is to extend the results in \cite{onlinetransportwa} to show that
the algorithm $\Greedy$ is constant competitive with resource augmentation $k \ge 3$. 
More specifically we show that

\begin{theorem}\label{thm:comp. general k}  For $k\ge 3$,
$\Greedy_k(I) \leq \left( 1 + \frac{2}{k-2} \right) \OPT(I)$.
\end{theorem}

Further we show that this bound is essentially tight by giving an instance where
this lower bound is obtained in the  limit.  So one possible interpretation of this
result is that $\Greedy$ should perform reasonably well (have bounded relative error) on 
instances where tripling the capacity of the garages wouldn't change the optimal cost
by more than a constant factor (so intuitively the load on the parking system is not too high).  It wouldn't be totally unreasonable to argue that this result 
provides a more convincing explanation of when $\Greedy$ should perform reasonably, and
why it performs reasonably in these instances, than do prior results. For example,
this result guarantees bounded competitiveness, and even competitiveness approaching one
as the resource augmentation increases. 
In fairness, let us acknowledge the best counterargument,
which is probably that a factor of three resource augmentation
is significant.

Not surprisingly, our proof of Theorem \ref{thm:comp. general k} builds on the foundation established in \cite{onlinetransportwa}.
However, it is important to note that if one naively applies the analysis of $\Greedy$ in \cite{onlinetransportwa} with $k \geq 3$ (instead of $k=2$), then one
just obtains $\log_k C$ competitiveness (instead of the original $\log_2 C$ competitiveness result). 
Thus we had to develop a new method to bound certain costs for
the $\Greedy$ algorithm. 
The main technical innovation was the introduction of what we call the the weighted tree cost.
Informally, the weighted tree cost bounds certain costs for
the $\Greedy$ algorithm by a particular weighted sum of the cost of certain edges in the optimal solution (instead of
directly bounding these costs by the entirety of the optimal cost). 

\section{Algorithm Analysis}

We begin with the simplifying assumption that $c_i = k$ and $a_i = 1$ for all $1 \leq i \leq n$. We assume the adversary services $r_i$ with $s_i$, and that the online algorithm services $r_i$ with $s_{\sigma(i)}$. By convention, we represent adversary edges by listing the request first (e.g. $(r_i, s_i)$) and online edges by listing the server first (e.g. $(s_{\sigma(i)}, r_i)$).

\subsection{Defining the Response Graph and Response Trees}

We start as in \cite{onlinetransportwa}
by defining the response graph,
noting that it is almost acyclic, 
and then decomposing  its edges into what we call response trees.
An example of a response graph and one possible decomposition into response trees can be seen in \Cref{fig:resp_graph} and \Cref{fig:sub5}.

\begin{definition}[Response Graph]
    Let $E_{\OPT} = \bigcup_{i=1}^n (r_i, s_i)$ be the set of all adversary edges,  $E_{ON} = \bigcup_{i=1}^n (s_{\sigma(i)}, r_i)$ be the set of all online edges, and  $E = E_{\OPT} \ \cup \ E_{ON}$. Then the \textit{response graph} is $\mathcal{G} = (S \cup R, E)$, where each edge has a weight that is the distance in the underlying metric space $\mathcal M$ between the endpoints of $e$. 
\end{definition}

\begin{figure}
\centering
\begin{subfigure}[t]{.4\textwidth}
  \centering
  \includegraphics[scale=.4]{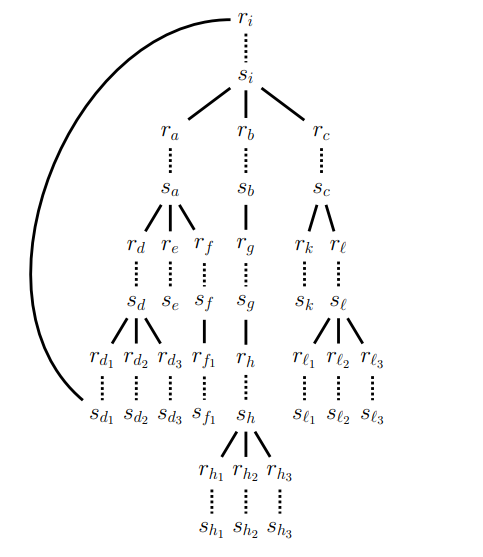}
    \caption{An Example Connected Component of the Response Graph for $k=3$.}
    \label{fig:resp_graph}
\end{subfigure}%
\quad
\begin{subfigure}[t]{.4\textwidth}
  \centering
  \includegraphics[scale=.43]{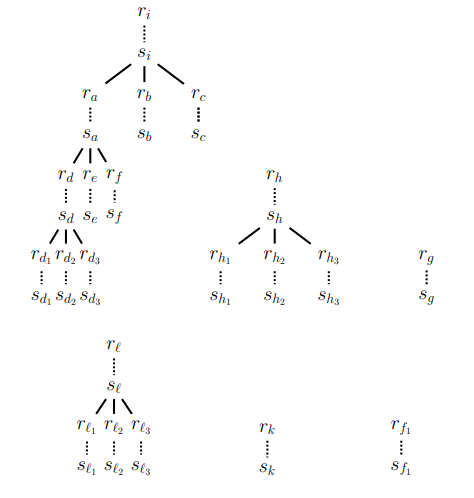}
    \caption{One possible tree decomposition (roots are highest node)}
    \label{fig:sub5}
\end{subfigure}%
\end{figure}

\begin{lemma}\label{lemma:2.4}\cite{onlinetransportwa}
Assume that request $r_i$ is in a cycle in $\mathcal{G}$. Then the connected component of $\mathcal{G} - (s_{\sigma(i)}, r_i)$ that contains $r_i$ is a tree.
\end{lemma}

\begin{definition}[Tree Decomposition]
    We define a tree decomposition of the response graph $\mathcal{G}$ to be 
a collection of response trees where:
\begin{itemize}
    \item Each response  tree $\mathcal T$ is a rooted tree that is rooted at some request $r_i$.
    \item Each response tree $\mathcal T$ is a subgraph of $\mathcal G$. 
    \item Every edge in $\mathcal{G}$ is
either contained in a unique response tree, or is the online edge $(s_{\sigma(i)}, r_i)$
incident to the root $r_i$ of some response tree  $\mathcal T$, but not both (so an online
edge incident to a root of a response tree is not in any response tree).
\end{itemize}
\end{definition}

\cite{onlinetransportwa} then shows how to decompose the response graph 
into response trees, where each response tree $ \mathcal{T}$  has the following additional properties:
\begin{itemize}
    \item  For each request $r_j \in \mathcal{T}$, $r_j$ has one  child, namely $s_j$.
    \item Each leaf  in $\mathcal{T}$ is a server site $s_j$ with parent $r_j$.
    
    \item Each nonleaf server site $s_i$ in $\mathcal T$ has $k$ incident online edges in $\mathcal T$, which are the children of $s_i$ in
    $\mathcal T$.
    \item For each request $r_j \in \mathcal{T}$ and for each leaf $s_q \in \mathcal{T}$ it is the case that the algorithm $\Greedy$ 
    had an unused server available at $s_q$ when request $r_j$ arrived. 
\end{itemize}
Intuitively, \cite{onlinetransportwa} accomplishes this by iteratively breaking
up each connected component $\mathcal{C}$ as follows.
Let $r_i$ be the most recent request in $\mathcal{C}$. First the
online edge $(s_{\sigma(i)}, r_i)$ is deleted.  Let $\mathcal{ C}'$ be the resulting connected component containing $r_i$
(note $\mathcal{C}'$ is a tree by \Cref{lemma:2.4}). 
A response tree rooted at $r_i$ is then created by including all vertices reachable from $r_i$ in $\mathcal{C}'$ by a
path that does not contain an unfull server site as an internal server site on the path (in this context, unfull means that at the time of $r_i$, the greedy algorithm had not used all of the
servers at that server site). Or alternatively, the leaves of $\mathcal T$ are unfull server sites reachable from $r_i$ in $\mathcal{C}'$ 
without passing through another unfull server site. The edges and request vertices of $\mathcal T$ are then removed from $\mathcal{C}$. We now fix a particular such
decomposition of $\mathcal G$ into response trees for the rest of the paper.

Finally, we give the following useful definitions related to response trees.

\begin{definition}[Adversary Cost of $\mathcal{T}$]
For a response tree $\mathcal T$, let $\OPT(\mathcal{T})$ be defined as the sum of the costs of all adversary edges in $\mathcal{T}$. Thus, $\OPT(\mathcal{T}) = \sum_{(r_j, s_j) \in \mathcal{T}} d(r_j, s_j)$.
    
\end{definition}

\begin{definition}[Online Cost of $\mathcal{T}$] 
For a response tree $\mathcal T$, 
let $\ON(\mathcal{T})$ be defined as the sum of the costs of all online edges in $\mathcal{T} \cup \{ (s_{\sigma(i)}, r_i) \}$. Thus, $\ON(\mathcal{T}) = \sum_{r_j \in \mathcal{T}} d(s_{\sigma(j)}, r_j)$.
\end{definition}

\begin{definition}[Leaf Distance in $\mathcal{T}$]
Let $\mathcal T$ be a response tree rooted at a request $r_i$. 
For all vertices $x \in \mathcal{T}$, let $\mathcal{T}(x)$ be the subtree of $\mathcal{T}$ rooted at $x$. Define $ld(x)$, the leaf distance of $x$, as the minimum distance in $\mathcal T$ (not in $\mathcal M$) from $x$ to a leaf of $\mathcal{T}(x)$. Further, define $ld(s_{\sigma(i)})$ 
to be $ d(s_{\sigma(i)}, r_i) + ld(r_i)$.
\end{definition}

\begin{definition}[Weighted Tree Cost]
Let $\mathcal T$ be a response tree rooted at a request $r_i$. 
Let $r_j$ be a request in $\mathcal T$, with child $s_j$ 
and grandchildren $r_{\delta(1)}, \ldots r_{\delta(k)}$. 
We then recursively define the weighted tree cost of $r_j$ to be
$$\displaystyle W(r_j) = d(r_j, s_j) + \frac{2}{k}\left(\sum_{h=1}^k W(r_{\delta(h)})\right)
$$
If $s_j$ is a leaf, then $W(r_j) = d(r_j, s_j)$.
\end{definition}

\subsection{Analysis of $\Greedy$}

In \Cref{lm:edge-bound} and \Cref{lm:ld-bound}, we show that the Weighted Tree Cost provides a useful upper bound on leaf distances of $\mathcal{T}$ and by extension online edges in $\mathcal{T}$. In \Cref{lm:bounded-tree}, we use this result to directly bound $\ON(\mathcal{T})$ in terms of $\OPT(\mathcal{T})$. Finally, in \Cref{lm:ai=1 case}, we extend this bound on trees $\mathcal{T}$ to the entire response graph $\mathcal{G}$, and finally prove the main result, \Cref{thm:comp. general k}.

\begin{lemma}\label{lm:edge-bound} 
Let $\mathcal T$ be an arbitrary response tree. 
For each request $r_j \in \mathcal{T}$, $d(s_{\sigma(j)}, r_j) \leq ld(r_j)$.
\end{lemma}
\begin{proof} 
Let $s_q$ be a leaf in the subtree of $\mathcal T$ rooted at $r_j$ that is closest to $r_j$ in $\mathcal T$. 
 By the triangle inequality, we know that $d(s_q, r_j) \leq ld(r_j)$. Further we know that $\Greedy$ had an unused 
 server at $s_q$ at the time that $r_j$ arrived. Thus by the definition of $\Greedy$, it must be the case that  $d(s_{\sigma(j)}, r_j) \leq d(s_q, r_j) $. Thus by transitivity we can conclude that $d(s_{\sigma(j)}, r_j) \leq ld(r_j)$.
\end{proof}

\begin{lemma}\label{lm:ld-bound} 
Let $\mathcal T$ be an arbitrary response tree. 
For each request $r_j \in \mathcal{T}$, $ld(r_j) \leq W(r_j)$.
\end{lemma}

\begin{proof}
We proceed by induction on the height $h$ of the induced tree $\mathcal{T}(r_j)$ (the subtree of $\mathcal{T}$ rooted at $r_j$). If $h=1$, the induced tree contains one request, $r_j$, which the adversary services using $s_j$. Thus, $ld(r_j) = d(r_j, s_j) = W(r_j)$.\par
Now suppose $h > 1$. Then, $s_j$ has $k$ children: $r_{b_1}, r_{b_2}, \dots, r_{b_k}$. This gives
\begin{align*}
ld(r_j) &= d(r_j,s_j) + \min_{p = 1 \dots k} \left[d(s_j, r_{b_p}) + ld(r_{b_p})\right] \\
&\leq d(r_j,s_j) + \min_{p = 1 \dots k} \left[ld(r_{b_p}) + ld(r_{b_p})\right] \\
&= d(r_j,s_j) + 2\left(\min_{p = 1 \dots k} \left[ld(r_{b_p})\right]\right)\\
&\leq d(r_j,s_j) + 2\left(\min_{p = 1 \dots k} \left[W(r_{b_p})\right]\right)\\
&\leq d(r_j,s_j) + \frac{2}{k}\left(\sum_{p=1}^k W(r_{b_p})\right)\\
&= W(r_j)
\end{align*} 
The first inequality follows from \Cref{lm:edge-bound}. The second inequality follows from induction.
The third inequality follows from the fact that the average has to be larger than the minimum.
The last equality follows from the definition of $W(r_j)$. 
\end{proof}

Lemma \ref{lm:ld-bound} is the main technical extension relative to \cite{onlinetransportwa}. In 
\cite{onlinetransportwa} the term $ld(r_j)$ is merely upper bounded by the optimal cost. If we used that
upper bound on $ld(r_j)$ in our analysis, our upper bound on the competitive ratio  would not be $O(1)$.

\begin{lemma}\label{lm:recursive-formula} 
Let $\mathcal T$ be an arbitrary response tree rooted at $r_i$ of height $h$. Let $D_{\ell}$ be the collection of adversary edges $(r_j, s_j)$ in $\mathcal T$ where the path from $r_i$ to $s_j$ in $\mathcal T$ passes
through $l$ adversary edges. Then
$$\displaystyle W(r_i) = \sum_{\ell=1}^h \left(\frac{2}{k}\right)^{l-1} \sum_{(r_j, s_j) \in D_{\ell}} d(r_j, s_j)
$$
\end{lemma}

\begin{proof}
We prove this by induction on $h$. For $h = 1$, we simply have $W(r_i) = d(r_i, s_i)$. For $h > 1$, suppose $s_i$ has $k$ incident online edges $(s_i, r_{a_1}), (s_i, r_{a_2})$, \dots, $(s_i, r_{a_k})$. Then $D_2 = \{(r_{a_1}, s_{a_1}), (r_{a_2}, s_{a_2}), \dots, (r_{a_k}, s_{a_k})\}$. Further, for all $2 \leq \ell \leq h$, define $D_{\ell}^{a_1} \subseteq D_{\ell}$ s.t. $D_{\ell}^{a_1}$ contains all adversary edges in $D_{\ell}$ within the subtree $\mathcal{T}(r_{a_1})$. Define $D_{\ell}^{a_2}, D_{\ell}^{a_3}, \dots, D_{\ell}^{a_k}$ similarly. Thus, $D_{\ell} = \bigcup_{p = 1}^k D_{\ell}^{a_p}$ where $D_{\ell}^{a_1}, D_{\ell}^{a_2}, \dots, D_{\ell}^{a_k}$ are pairwise disjoint. Then we have
\begin{align*}
        W(r_i)
        & = d(r_i, s_i) + \frac{2}{k}\left(\sum_{p=1}^k W(r_{a_p})\right)\\
        & = d(r_i, s_i) + \frac{2}{k}\left(\sum_{p=1}^k \left(\sum_{l=1}^{h-1} \left(\frac{2}{k}\right)^{l-1} \sum_{(r_j, s_j) \in D_{l+1}^{a_p}} d(r_j, s_j)\right)\right) \\
        & = d(r_i, s_i) + \frac{2}{k}\left(\sum_{l=1}^{h-1} \left(\frac{2}{k}\right)^{l-1} \left(\sum_{p=1}^k \sum_{(r_j, s_j) \in D_{l+1}^{a_p}} d(r_j, s_j)\right)\right) \\
        & = d(r_i, s_i) + \frac{2}{k}\left(\sum_{l=1}^{h-1} \left(\frac{2}{k}\right)^{l-1} \sum_{(r_j, s_j) \in D_{l+1}} d(r_j, s_j)\right) \\
        & = d(r_i, s_i) + \sum_{l=2}^h
        \left(\frac{2}{k}\right)^{l-1} \sum_{(r_j, s_j) \in D_{\ell}} d(r_j, s_j) \\
        & = \sum_{l=1}^h \left(\frac{2}{k}\right)^{l-1} \sum_{(r_j, s_j) \in D_{\ell}} d(r_j, s_j) \\
\end{align*}
\end{proof}

\begin{lemma}\label{lemma 5}\label{lm:bounded-tree}
Let $\mathcal T$ be an arbitrary response tree.
Then $\ON(\mathcal{T}) \leq \left(1 + \frac{2}{k-2}\right)\OPT(\mathcal{T})$.
\end{lemma}
\begin{proof}
Using \Cref{lm:edge-bound} and \Cref{lm:ld-bound},  we can bound $\ON(\mathcal{T})$ as follows:
    
    \[ \ON(\mathcal{T}) = \sum_{r_j \in \mathcal{T}} d(s_{\sigma(j)}, r_j) \leq \sum_{r_j \in \mathcal{T}} ld(r_j) \leq \sum_{r_j \in \mathcal{T}} W(r_j) \]

Note by applying \Cref{lm:recursive-formula}  one can view $\sum_{r_j \in \mathcal{T}} W(r_j)$ as a linear combination of costs of adversary edges. 
Consider an arbitrary adversary edge $(r_q, s_q) \in \mathcal{T}$. Note again by \Cref{lm:recursive-formula} that $d(r_q, s_q)$ will be included in $W(r_j)$ with coefficient $\left(\frac{2}{k}\right)^{b-1}$ only when the following two conditions hold: $r_j$ is an ancestor of $r_q$, and the simple path from $r_j$ to $s_q$ in $\mathcal T$ passes through $b$ adversary edges. Further, clearly an ancestor $r_j$ which satisfies these conditions is unique. Thus the coefficient associated with the cost of $(r_q, s_q)$ in  $\sum_{r_j \in \mathcal{T}} W(r_j)$ is at most
$1 + \left(\frac{2}{k}\right) + \left(\frac{2}{k}\right)^2 + \dots = \frac{1}{1-\frac{2}{k}} = \frac{k}{k-2} = 1 + \frac{2}{k-2}$. Thus, we have 
\[  \sum_{r_j \in \mathcal{T}} W(r_j) \leq \left(1 + \frac{2}{k-2}\right) \sum_{(r_q, s_q) \in \mathcal{T}} d(r_q, s_q) = \left(1 + \frac{2}{k-2}\right) \OPT(\mathcal{T}) \]

\end{proof}

\begin{lemma}\label{lm:ai=1 case}
 $\Greedy_k(I) \leq \left( 1 + \frac{2}{k-2} \right) \OPT(I)$  under the assumption that each server site $s_i$ has $k \geq 3$ online servers and one adversary server.
\end{lemma}

\begin{proof}
 Note that $\Greedy_k(I)$ is equal to the total cost of the online edges in the response graph, $\mathcal G$. Via our tree decomposition, this cost is
 
\[ \Greedy_k (I) = \sum_{\mathcal{T} \in \mathcal{G}} \ON(\mathcal{T}) \leq \left(1 + \frac{2}{k-2}\right)\sum_{\mathcal{T} \in \mathcal{G}} \OPT(\mathcal{T}) = \left(1 + \frac{2}{k-2}\right) \OPT(I) \]

\end{proof}

\begin{proof}[Proof of \Cref{thm:comp. general k}]
    Split each server site with online capacity $c_i= k a_i$ and adversary capacity $a_i $ into $a_i$ server sites with online capacity $k$ and adversary capacity 1. Clearly the optimal cost is the same because the underlying server locations have not changed. Further, $\Greedy_k$ assigns requests identically on both instances for the same reason, and so the online cost is the same as well. Thus \Cref{lm:ai=1 case} directly gives the desired result.

\end{proof}

\subsection{Algorithm Lower Bound} 

Lastly, we show the competitiveness bound of $1 + \frac{2}{k-2}$ for $\Greedy_k$ is essentially tight.

\begin{theorem}\label{thm:kweak_lowerbound}
    $\forall \ \epsilon > 0$, there is an instance $I_{\epsilon}$ where $\Greedy_k(I_{\epsilon}) > \left( 1 + \frac{2}{k-2}  - \epsilon \right) \OPT(I_{\epsilon})$. 
\end{theorem}
\begin{proof}
    We embed $m$ server sites on the real line. The server site $s_1$ is located at the point $-1$. For $2 \leq i \leq m$, the server site $s_i$ is located at $2^{i-1}-1$. The online algorithm has $c_i = k^{m-i+1}$ servers at site $s_i$, and the adversary has $a_i = k^{m-i}$ servers at site $s_i$. The requests occur in $m$ batches. The first batch consists of $k^{m-1}$ requests at 0. For $2 \leq i \leq m$, the $i$-th batch consists of $k^{m-i}$ requests at $s_i = 2^{i-1} - 1$. $\Greedy_k$ responds to batch $i \ (1\leq i<m)$ by answering each request in batch $i$ with server site $s_{i+1}$, thus depleting $s_{i+1}$. $\Greedy_k$ responds to batch $m$ by answering the sole request with site $s_1$.

    For batch 1, $\Greedy_k$ services $k^{m-1}$ requests, each of which requires a cost of $s_2 - 0 = 1$. For batch $i$, $1 < i < m$, $\Greedy_k$ services $k^{m-i}$ requests, each of which requires a cost of $s_{i+1} - s_i = (2^i - 1) - (2^{i-1} - 1) = 2^{i-1}$. For batch $m$, $\Greedy_k$ services 1 request, which requires a cost of $s_m - s_1 = (2^{m-1} - 1) - (-1) = 2^{m-1}$. Thus $\Greedy_k$ incurs a total cost of 
    \begin{align*}
        \Greedy_k(I_{\epsilon}) = k^{m-1} \cdot 1 + \sum_{i=2}^{m-1} k^{m-i} \cdot 2^{i-1} + 1 \cdot 2^{m-1}
        &= \sum_{i=1}^m k^{m-i} \cdot 2^{i-1} \\
        & = k^{m-1} \sum_{i=1}^m \frac{2^{i-1}}{k^{i-1}} \\
        &= k^{m-1} \sum_{i=1}^m \left(\frac{2}{k}\right)^{i-1} \\
        &= k^{m-1} \left(\frac{1 - \left(\frac{2}{k}\right)^m}{1 - \frac{2}{k}}\right) \\
        &= k^{m-1} \cdot \frac{k}{k-2}\left(1 - \left(\frac{2}{k}\right)^m\right) \\
        &= k^{m-1} \cdot \left(1 + \frac{2}{k-2}\right)\left(1 - \left(\frac{2}{k}\right)^m\right)
    \end{align*} 

    The adversary could respond to the requests by servicing batch $i$ with server site $s_i$. The adversary would incur a cost of $k^{m-1}$ for batch 1, and a cost of 0 for batches $i$, $2 \leq i \leq m$. Then the adversary can achieve a total cost of $\OPT(I_{\epsilon}) \leq k^{m-1}$. Thus, we have
    \[
    \frac{\Greedy_k(I_{\epsilon})}{\OPT(I_{\epsilon})} \geq \frac{k^{m-1} \cdot \left(1 + \frac{2}{k-2}\right)\left(1 - \left(\frac{2}{k}\right)^m \right)}{k^{m-1}} = \left(1 + \frac{2}{k-2}\right)\left(1 - \left(\frac{2}{k}\right)^m\right)
    \]
    
    For any $\epsilon > 0$, for sufficiently large $m$, 
    $\left(\frac{2}{k}\right)^m < \left(\frac{k-2}{k}\right)\epsilon$ giving 
    \[
    \left(1 + \frac{2}{k-2}\right) \left(1 - \left(\frac{2}{k}\right)^m\right) > \left(1 + \frac{2}{k-2}\right)\left(1 - \left(\frac{k-2}{k}\right)\epsilon\right) = 1 + \frac{2}{k-2} - \epsilon
    \] 
    
    Thus for sufficiently large $m$,
    \[
    \Greedy_k(I_{\epsilon}) > \left( 1 + \frac{2}{k-2}  - \epsilon \right) \OPT(I_{\epsilon})
    \]
    \end{proof}

\bibliographystyle{abbrv}
\bibliography{bib.bib}

\end{document}